\documentclass[reqno, centertags, 11pt]{amsart}
\usepackage{amsmath,amsthm,amsfonts,amssymb}
\usepackage{graphicx}
\usepackage{epstopdf}

\renewcommand{\part}[1]{\left(#1\right)}  %--round
\newcommand{\ket}[1]{\vert {#1} \rangle}  %---use: \ket{\psi}
\newcommand{\nl}{\par\noindent}   %---new line without indentation
\newtheorem{theorem}{Theorem}[section]
\newtheorem{lemma}{Lemma}[section]

\newtheorem{definition}{Definition}[section]

\numberwithin{equation}{section}

{\begin{proof}[#1]%
\renewcommand{\qed}{\rule{0mm}{1mm}\hfill{\footnotesize\textsc{q.e.d.}\
(#2)}}}
{\end{proof}
}
\begin{document}
   \title[Epistemic
    semantics  ]
 {\textbf{A quantum computational semantics for
    epistemic logical operators. \nl Part II: Semantics
}}

\author[Beltrametti]{Enrico Beltrametti}
\address[Beltrametti]{Dipartimento di Fisica,
Universit\`a di Genova,
  Via Dodecaneso, 33,
I-16146 Genova, Italy} \email{enrico.beltrametti@ge.infn.it}

\author[Dalla Chiara]{Maria Luisa Dalla Chiara}
\address[Dalla Chiara]{Dipartimento di Filosofia,
    Universit\`a di Firenze,
    Via Bolognese 52, I-50139 Firenze, Italy}
\email{dallachiara@unifi.it}

%-
\author[Giuntini]{Roberto Giuntini}
\address[R.~Giuntini]{Dipartimento di Filosofia e Teoria
 delle Scienze Umane,
    Universit\`a di Cagliari,
    Via Is Mirrionis 1, I-09123 Cagliari, Italy.}
    \email{giuntini@unica.it}

%-

\author[Leporini]{Roberto Leporini}
\address[R. Leporini]{Dipartimento di Ingegneria,
    Universit\`a di Bergamo,
    viale Marconi 5, I-24044 Dalmine (BG), Italy.}
\email{roberto.leporini@unibg.it}

%-
\author[Sergioli]{Giuseppe Sergioli}
\address[G.~Sergioli]{Dipartimento di Filosofia e Teoria
 delle Scienze Umane,
    Universit\`a di Cagliari,
    Via Is Mirrionis 1, I-09123 Cagliari, Italy.}
    \email{giuseppe.sergioli@gmail.com}

%--------
%\dedicatory{Dedicated to Gianpiero Cattaneo}
%\thanks{Corresponding author: R.~Giuntini,  \rm{giuntini@unica.it}}

\maketitle

\begin{abstract}By using the abstract structures investigated in
the first Part of  this article, we develop a semantics for  an
epistemic language, which  expresses sentences like ``Alice knows
that Bob does not understand that $\pi$ is irrational''. One is
dealing with a {\em holistic\/} form of  quantum computational
semantics, where entanglement plays a fundamental role; thus,  the
{\em meaning\/} of a global expression determines the {\em
contextual meanings\/} of its parts, but generally not the other
way around. The epistemic situations represented in this semantics
seem to reflect some characteristic limitations of the real
processes of acquiring information. Since knowledge  is not
generally closed under logical consequence, the unpleasant
phenomenon of {\em logical omniscience\/} is here avoided.

\end{abstract}

\section{The epistemic quantum computational syntax}
The  structures, investigated in the first Part  of this article,
provide the mathematical basis for the development of our
epistemic semantics. The basic intuitive idea can be sketched as
follows: pieces of quantum information (qumixes) can be denoted by
the sentences of a formal language, whose logical connectives
correspond to some quantum  gates, while the two fundamental
epistemic operators ({\em to understand\/} and  {\em to know\/})
are interpreted  as epistemic operations living in semantic models
based on convenient epistemic quantum computational structures.
Accordingly, sentences like ``At time $t$ Alice knows that Bob
does not understand that $\pi$ is irrational'' turn out to denote
particular examples  of qumixes, representing possible states of
quantum objects.
Let us first introduce the epistemic language that will be used.
This language, indicated by $\mathcal L^{EpQC}$, contains atomic
sentences (say, ``the spin-value in the $x$-direction is up''),
including two privileged sentences $\bf t$ and $\bf f$ that
represent the truth-values {\em Truth\/} and {\em Falsity\/},
respectively. We will  use $\bf q, \bf r,\ldots$ as metavariables
for atomic sentences, and $\alpha$, $\beta$, ... as metavariables
for sentences. The quantum computational connectives of $\mathcal
L^{EpQC}$ are: the negation $\lnot$ (which corresponds to the gate
{\em Negation\/}), the squareroot of the identity $\sqrt{id}$
(which corresponds to {\em Hadamard\/}), a ternary connective
$\intercal$ (which corresponds to the  {\em Toffoli\/}-gate), the
exclusive disjunction $\biguplus$ (which corresponds to ${\tt
XOR}$). The epistemic sub-language of $\mathcal L^{EpQC}$
contains: a finite set of names for epistemic agents ($\mathbf a$,
$\mathbf b$, ...); a set of names ($t_1$, $t_2$, ...) for the
elements of a given time-sequence; the epistemic operators
$\mathcal U$ (to understand) and $\mathcal K$ (to know).

 For any sentences
$\alpha$, $\beta$, $\gamma$, the expressions $\lnot \alpha$,
$\sqrt{id}\, \alpha$, $\intercal(\alpha, \beta,\, \gamma)$,
$\alpha\biguplus \beta$ are sentences. For any sentence $\alpha$,
for any agent-name $\bf a$ (say, Alice), for any time-name $t$,
the expressions $\mathcal U\mathbf {a}_t\alpha$ ({\em at time $t$
agent $\mathbf a$ understands the sentence $\alpha$}) and
$\mathcal K\mathbf {a}_t\alpha$ ({\em at time $t$ agent $\mathbf
a$ knows the sentence $\alpha$}) are sentences. Notice that nested
epistemic operators are allowed: an expression like $\mathcal
K\mathbf a_t\lnot \mathcal U\mathbf b_t \alpha$ is a well-formed
sentence.
 The
connectives $\lnot$, $\sqrt{id}$, $\intercal$, $\biguplus$ are
called {\em gate-connectives\/}. Any subexpression $\mathcal
U\mathbf a_ t$ or $\mathcal K\mathbf a_t$ of an  epistemic
sentence will be  called an {\em epistemic connective}.

We recall that, for any truth-perspective $\mathfrak T$, the
Toffoli-gate permits one to define a reversible conjunction ${\tt
AND}_\mathfrak T$ (for any  $\rho \in \mathfrak D(\mathcal
H^{(m)})$ and for any $\sigma \in \mathfrak D(\mathcal H^{(n)})$):
$${\tt AND}_\mathfrak T(\rho,\sigma):= \,\,
^\mathfrak D{\tt T}_\mathfrak T^{(m,n,1)} (\rho \otimes \sigma
\,\otimes \,^\mathfrak TP_0^{(1)}).   $$ Accordingly, from a
syntactical point of view, it is reasonable to define
(metalinguistically)  the logical conjunction $\land$ as follows
(for any sentences $\alpha$ and $\beta$):
$$ \alpha \land \beta := \intercal(\alpha, \beta, \mathbf f).  $$

 We will now  introduce some syntactical
 notions that will be used in our semantics.

 \begin{definition}\label{de:gatesent}\nl

 \begin{itemize}
 \item $\alpha$ is called a gate-sentence iff either $\alpha$
     is atomic or the principal connective of  $\alpha$ is a
     gate-connective.
 \item $\alpha$ is called an epistemic sentence iff $\alpha$
     has the form $\mathcal K\mathbf {a}_t\beta$.

 \end{itemize}

 \end{definition}

 \begin{definition} {\em (The atomic complexity of a sentence)}  \label{de:atcomp}\nl
 The atomic complexity $At(\alpha)$ of a sentence $\alpha$ is
  the number of occurrences of atomic sentences in
 $\alpha$.
 \end{definition}

 For instance, $At(\intercal(\mathbf q, \mathbf q,\mathbf f)) = 3$.
 We will also indicate by $\alpha^{(n)}$ a sentence whose atomic
 complexity is $n$. The notion of atomic complexity plays an
 important semantic  role. As we will see, the meaning of any sentence whose
 atomic complexity is $n$ shall live in the domain
 $\mathfrak D(\mathcal H^{(n)})$. For this reason,
 $\mathcal H^{(At(\alpha))}$ (briefly indicated by $\mathcal
 H^\alpha$) will be also called the {\em semantic space\/} of $\alpha$.

 Any sentence
$\alpha$ can be naturally decomposed into its parts, giving rise
to a special configuration called the {\em syntactical tree\/} of
$\alpha$ (indicated by {$STree^{\alpha}$}).

Roughly, $STree^{\alpha}$ can be represented as a finite sequence
of {\em levels\/}:
$$Level_k(\alpha)$$
$$\vdots$$
$$Level_1(\alpha),$$
where:

\begin{itemize}
\item each $Level_i(\alpha)$ (with $1 \le i \le k$) is a
    sequence   $(\beta_1, \ldots,  \beta_m)$  of subformulas
    of $\alpha$;
    \item the {\em bottom level\/} $Level_1(\alpha)$ is
        $\alpha$; \item the {\em top level\/ }
        $Level_k(\alpha)$ is the sequence $(\mathbf q_1,
        \ldots, \mathbf q_r)$, where $\mathbf q_1, \ldots
        \mathbf q_r$ are the atomic occurrences in $\alpha$;
        \item for any $i$ (with $1 \le i < k$),
            $Level_{i+1}(\alpha)$ is the  sequence obtained by
            dropping the {\em principal gate-connective\/} in
            all molecular gate-sentences  occurring at
            $Level_i(\alpha)$, by dropping the epistemic
            connectives ($\mathcal U\mathbf {a}_t$, $\mathcal
            K\mathbf {a}_t$) in all epistemic sentences
            occurring at $Level_i(\alpha)$ and by repeating
            all the atomic sentences that  occur at
            $Level_i(\alpha)$.
\end{itemize}

By {\em Height\/} of $\alpha$ (indicated by ${Height(\alpha)}$) we
mean the number of levels of the syntactical tree of $\alpha$.

As an example, consider the following sentence: $$\alpha= \mathcal
K\mathbf{a}_t\lnot(\mathbf q \land \lnot \mathbf q) = \mathcal K
\mathbf{a}_t\lnot(\intercal(\mathbf q,\lnot \mathbf q,\mathbf
f))$$ (say, ``At time $t$ Alice knows the non-contradiction
principle'', instantiated by means of the atomic sentence $\bf
q$).

The syntactical tree of $\alpha$ is the following sequence of
levels:
\begin{align*}
Level_5(\alpha) &= (\mathbf q, \mathbf q, \mathbf f)\\
Level_4(\alpha) &= (\mathbf q,\lnot \mathbf q;\mathbf f)\\
Level_3(\alpha) &=
(\intercal(\mathbf q,\lnot \mathbf q,\mathbf f))\\
Level_2(\alpha) &=
(\lnot\intercal(\mathbf q,\lnot \mathbf q,\mathbf f))\\
Level_1(\alpha) &= (\mathcal K\mathbf{a}_t
\lnot(\intercal(\mathbf q,\lnot \mathbf q,\mathbf f)))\\
\end{align*}
Clearly, $Height(\intercal(\mathbf q, \lnot \mathbf q, \mathbf
f))= 5.$

More precisely, the syntactical tree of a sentence (whose atomic
complexity is $r$) is defined as follows.

\begin{definition} {\em (The syntactical tree of $\alpha$)} \label{de:sintree}\nl
The syntactical tree of $\alpha$ is the following sequence of
sentence-sequences:
$$STree^\alpha  = (Level_1(\alpha),\ldots, Level_k(\alpha)), $$
where:
\begin{itemize}

\item $Level_1(\alpha)= \alpha$;

\item $Level_{i+1}$ is defined as follows for any $i$ such
    that $1 \le i < k$. The following cases are possible:
\begin{enumerate} \item $Level_i(\alpha)$ does not contain any
 connective. Hence, $Level_i(\alpha)= (\mathbf q_1,
\ldots,  \mathbf q_r)$ and $Height(\alpha) = i$;
\item $Level_i(\alpha)=   (\beta_1, \ldots, \beta_m)$, and
    for at least one $j$,  $\beta_j$ has a (principal)
    connective. Consider the following sequence of
    sentence-sequences:
$$\mathcal s^\prime_1, \ldots, \mathcal s^\prime_m,  $$ where:
$ \mathcal s^\prime_h = \begin{cases}  (\beta_h),\,
 \text{if} \,\,\beta_h \,\,\text{is atomic};  \\
                     \mathcal s_h^*, \,\,\text{otherwise.}
                     %\end{cases}$

                     \text{Where: }\nl

                     \mathcal s_h^*= \begin{cases}
                     (\delta),\, \text{if}\,\, \beta_h =
                     \lnot\delta\,\,
                      \text{or} \,\,
                     \beta_h = \sqrt{id}\,\,\delta;  \\
                     (\gamma, \delta, \theta),
                      \,\,\text{if} \,\, \beta_h =
                      \intercal(\gamma,\delta,\theta);\\
(\gamma,\delta),\,\,\text{if} \,\, \beta_h = \gamma
                      \biguplus \delta;\\
(\delta),\, \text{if}\,\, \beta_h = \mathcal U\mathbf
{a}_t\delta \,\,\text{or}\,\,  \beta_h = \mathcal K\mathbf
{a}_t\delta.
                     \end{cases}

                     \end{cases}$

Then, $$Level_{i+1}(\alpha) =  \mathcal s^\prime_1 \bullet
\ldots \bullet \mathcal s^\prime_m,$$ where $\bullet$
        represents
the sequence-composition.
\end{enumerate}
\end{itemize}

\end{definition}

\section{The epistemic quantum computational semantics}
 We will now give the basic definitions of our  semantics.
We will apply to epistemic situations  a  {\em holistic\/} version
of quantum computational semantics (which has been naturally
inspired by the characteristic holistic features of the quantum
theoretic formalism\footnote{See \cite{DFGLLS} and \cite{DBGS}.}).
In this semantics any model assigns to any sentence a {\em global
meaning\/} that determines the contextual meanings of all its
parts (from the whole to the parts!). It may happen that one and
the same sentence receives different meanings in different
contexts.

%---------------------

%----------------------
Before defining  holistic  models, we will first introduce the
weaker notion of {\em quasi-model\/} of the language $\mathcal
L^{EpQC}$.

\begin{definition}{\em (Quasi-model)}\nl
 A {\em quasi-model\/} of the language $\mathcal L^{EpQC}$ is a system

 $$\mathcal M^q  = (T,\,Ag,\, \mathbf {EpSit}, \,den)$$
 where:
 \begin{enumerate}
 \item $(T,\,Ag,\, \mathbf {EpSit})$ is an epistemic quantum
     computational structure\footnote{See Section 3 of the
     first Part of this article.};
 \item $den$ is a function that interprets the individual
     names of the language. By simplicity, we put:
     $$den({\bf a})= \frak a;  \,\,\, den(t)= \frak t.   $$

  \end{enumerate}
  \end{definition}

  Apparently, quasi-models represent partial interpretations of the
  language: while names of times and of agents  receive an
  interpretation in the framework of a given epistemic situation,
  meanings of sentences are not determined.

  In the first Part of this article we have seen that knowledge
  operations cannot be generally represented as qumix gates. At
  the same time, once fixed an epistemic quantum computational
  structure $\mathcal S  = (T,\,Ag,\, \mathbf {EpSit})$, one can
  naturally define the following notion of {\em pseudo-gate\/}
  with respect to $\mathcal S$.

  \begin{definition} {\em (Pseudo-gate)}\label{de:pseudo}\nl Let
$\mathcal S  = (T,\,Ag,\, \mathbf {EpSit})$ be an epistemic
quantum computational structure. A {\em pseudo-gate\/} of
$\mathcal S$ is a operator-product
$$\mathbf X_1^{(n_1)} \otimes \ldots \otimes  \mathbf X_m^{(n_m)}, $$
where any $\mathbf X_i^{(n_i)}$ (with $1 \le i \le m$) is either a
qumix-gate $^\mathfrak D G_\mathfrak T^{(n_i)}$ with respect to a
 truth-perspective $\mathfrak T$ or an epistemic operation
 ($\mathbf U_{\mathfrak {a_t}}^{(n_i)}$ or
 $\mathbf K_{\mathfrak {a_t}}^{(n_i)}$) of $\mathcal S$.

  \end{definition}

  One can show that for any choice of a truth-perspective  $\frak T$  and of a
  quasi-model
  $\mathcal M^q  = (T,\,Ag,\, \mathbf {EpSit}, \,den)$, the
  syntactical tree of a sentence $\alpha$ uniquely determines a
  sequence of pseudo-gates, that will be called the
  {\em ($\mathfrak T,\mathcal M^q$)-pseudo-gate tree\/} of $\alpha$.

As an example, consider again the sentence
$$\alpha= \mathcal
K\bf{a_t}\lnot(\mathbf q \land \lnot \mathbf q) = \mathcal K
\bf{a_t}\lnot(\intercal(\mathbf q,\lnot \mathbf q,\mathbf f))$$
 and its syntactical tree.

  Apparently, $Level_4(\alpha)$ is
obtained from $Level_5(\alpha)$ by repeating the first occurrence
of $\mathbf q$, by negating the second occurrence of $\mathbf q$
and by repeating $\mathbf f$. Hence the pseudo-gate that
corresponds to $Level_4(\alpha)$ will  be $^\mathfrak D{\tt
I}^{(1)} \otimes\, ^\mathfrak D{\tt NOT}^{(1)}_\mathfrak T \otimes
\,^\mathfrak D{\tt I}^{(1)}$. $Level_3(\alpha)$ is obtained from
$Level_4(\alpha)$ by applying to the three sentences occurring at
$Level_4(\alpha)$ the connective $\intercal$. Hence the
pseudo-gate that corresponds to $Level_3(\alpha)$ will  be
$^\mathfrak D{\tt T}^{(1,1,1)}_{\mathfrak T}$. $Level_2(\alpha)$
is obtained from $Level_3(\alpha)$ by applying to the  sentence
occurring at $Level_3(\alpha)$ the connective $\lnot$. Hence the
pseudo-gate that corresponds to $Level_2(\alpha)$ will  be
$^\mathfrak D{\tt NOT}^{(3)}_\mathfrak T$. Finally,
$Level_1(\alpha)$ is obtained from $Level_2(\alpha)$ by applying
to the  sentence occurring at $Level_2(\alpha)$ the epistemic
connective $\mathcal K\mathbf a_t$. Hence the pseudo-gate that
corresponds to $Level_1(\alpha)$ will  be $\mathbf
K^{(3)}_{\mathfrak{a_t}}$.

On this basis, the {\em ($\mathfrak T,\mathcal M^q$)-pseudo-gate
tree\/}
 of the sentence
 $$\alpha= \mathcal
K\bf{a_t}\lnot(\mathbf q \land \lnot \mathbf q) = \mathcal K
\bf{a_t}\lnot(\intercal(\mathbf q,\lnot \mathbf q,\mathbf f))$$
  can be identified with  the following sequence consisting of
  four
pseudo-gates:
$$\left(^\mathfrak D{\tt I}^{(1)} \otimes\,
^\mathfrak D{\tt NOT}^{(1)}_\mathfrak T
\otimes \,^\mathfrak D{\tt I}^{(1)},\,\,\,
^\mathfrak D{\tt T}_\mathfrak T^{(1,1,1)}, \,\,\,
^\mathfrak D{\tt NOT}^{(3)}_\mathfrak T,\,\,\,
\mathbf K^{(3)}_{{\mathfrak a_t}}\right).$$

Notice that the  truth-perspectives $\mathfrak T$ and $\mathfrak
{T_{a_t}}$ may be different.

The general definition of {\em ($\mathfrak T,\mathcal
M^q$)-pseudo-gate tree\/}
 is the following:

\begin{definition} {\em ({\em ($\mathfrak T,\mathcal M^q$)-pseudo-gate
tree\/})} \label{de:qumtree} \nl Let $\alpha$ be a sentence such
that $Height(\alpha)= k$. The {\em ($\mathfrak T,\mathcal
M^q$)-pseudo-gate tree\/} of $\alpha$ is the sequence of
peudo-gates
$$ PsTree_\mathfrak T^\alpha =
 (^\alpha\mathbf O^{(k-1)}_\mathfrak T,\ldots,
 \, ^\alpha\mathbf O^{(1)}_\mathfrak T),  $$
that is defined as follows. Suppose that
$$Level_{i-1}(\alpha) = (\beta_1^{(r_1)}, \ldots,
\beta_m^{(r_m)}),$$ (where $1 < i \le k$). We put:
$$ ^\alpha\mathbf O^{(i-1)}_\mathfrak T=
\,^\alpha\mathbf X_{\mathfrak T}^{(r_1)} \otimes \ldots
 \otimes \, ^\alpha \mathbf X_{\mathfrak T}^{(r_m)},$$
where any $^\alpha\mathbf X_{\mathfrak T}^{(r_j)}$ is a
pseudo-gate defined on $\mathcal H^{(r_j)}$ such that:

$^\alpha\mathbf X_{\mathfrak T}^{(r_j)}= \begin{cases} ^\mathfrak
D{\tt I}^{(r_j)}, \,\, \text{if}\,\, \beta_j^{(r_j)}
\,\, \text{is atomic};\\

               ^\mathfrak D{\tt NOT}_\mathfrak T^{(r_j)}, \,\, \text{if}
               \,\,\beta_j^{(r_j)}= \lnot \delta;\\
               ^\mathfrak D\sqrt{{\tt I}}_\mathfrak T^{(r_j)}, \,\, \text{if}
               \,\,\beta_j^{(r_j)}= \sqrt{id}\,\, \delta;\\
               ^\mathfrak D{\tt T}_\mathfrak T^{(u,v,w)}, \,\, \text{if}
               \,\,\beta_j^{(r_j)}=
\intercal (\gamma^{(u)}, \delta^{(v)}, \theta^{(w)});\\
^\mathfrak D{\tt XOR}_\mathfrak T^{(u,v)}, \,\, \text{if}
               \,\,\beta_j^{(r_j)}=
\gamma^{(u)} \biguplus \delta^{(v)}; \\
\mathbf U^{(r_j)}_{\mathfrak{a_t}},\,\,\text{if}\,\,
\beta_j^{(r_j)}= \mathcal U\mathbf a_t\delta;\\
 \mathbf
K^{(r_j)}_{\mathfrak{a_t}},\,\,\text{if}\,\, \beta_j^{(r_j)}=
\mathcal K\mathbf a_t\delta.
               \end{cases}$

\end{definition}

Consider now a sentence  $\alpha$ and let $(^\alpha\mathbf
O^{(k-1)}_\mathfrak T,\ldots,
 \, ^\alpha\mathbf O^{(1)}_\mathfrak T)$ be the
($\mathfrak T,\mathcal M^q$)-pseudo-gate tree of $\alpha$. Any
choice of a qumix $\rho$ in $\mathcal H^\alpha$ determines a
sequence $(\rho_k,\ldots,\rho_1)$ of qumixes of $\mathcal
H^\alpha$, where:
$$\rho_k=\rho$$
$$\rho_{k-1}= \,\,^\alpha\mathbf O_{\mathfrak T}^{k-1}(\rho_k)$$ $$\vdots $$
 $$\rho_{1}=
\,\,^\alpha\mathbf O_{\mathfrak T}^{1}(\rho_2).$$ The qumix $\rho_k$ can
be regarded as a possible {\em input-information\/} concerning the
atomic parts of $\alpha$, while $\rho_1$  represents the {\em
output-information\/} about $\alpha$, given the input-information
$\rho_k$. Each $\rho_i$ corresponds to the {\em information} about
$Level_i(\alpha)$, given the input-information $\rho_k$.

How to determine an information about the parts of $\alpha$ under
a given input? It is natural to apply the  {\em reduced state
function\/} that  determines for any state $\rho$ of a composite
system $S= S_1 + \ldots + S_n$ the state
$Red^{i_1,\ldots,i_m}(\rho)$ of any subsystem $S_{i_1}+ \ldots +
S_{i_m}$ (where $1\le i_1 \le n, \ldots, 1\le i_m \le n $.)
Consider the syntactical tree of $\alpha$ and suppose that:
$$ Level_i(\alpha)= (\beta_{i_1}, \ldots, \beta_{i_r}).$$
We know that  the ($\mathfrak T,\mathcal M^q$)-pseudo-gate tree of
$\alpha$ and the choice of an input $\rho_k$ (in $\mathcal
H^\alpha$) determine a sequence of qumixes:
 $$\rho_k \leftrightsquigarrow Level_k(\alpha) = (\mathbf
q_1, \ldots, \mathbf q_t)$$ $$\vdots$$ $$ \rho_i
\leftrightsquigarrow Level_i(\alpha) = (\beta_{i_1}, \ldots,
\beta_{i_r}) $$ $$\vdots $$
$$ \rho_1 \leftrightsquigarrow Level_1(\alpha) = \alpha$$

We can consider ${ Red^j(\rho_i)}$, the {\em reduced information
of $\rho_i$ with respect to the $j$-th part}. From a semantic
point of view, this object can be regarded as a {\em contextual
information} about $\beta_{i_j}$ (the subformula of $\alpha$
occurring at the $j$-th position at $Level_i(\alpha)$) under the
input $\rho_k$.

We can now define the notion of {\em holistic model\/}, which
assigns meanings to all sentences of the language, for any choice
of a truth-perspective $\mathfrak T$.

%---------------------
\begin{definition}{\em (Holistic model)} \label{de:hol}\nl
A {\em holistic model} of the language $\mathcal L^{EpQC}$ is a
system
$$\mathcal M  = (T,\,Ag,\, \mathbf {EpSit}, \,den,
\, ^\mathcal M{\tt Hol})$$
 where:
 \begin{enumerate}
\item[(1)] $(T,\,Ag,\, \mathbf {EpSit}, \,den)$ is a
    quasi-model $\mathcal M^q$ of the language. \item [(2)]
    $^\mathcal M{\tt Hol}$ is a map that associates to any
    truth-perspective $\mathfrak T$ a map $^\mathcal M{\tt
    Hol}_\mathfrak T$  representing a {\em holistic
    interpretation\/}  of the sentences of the language. The
    following conditions are required.
    \begin{enumerate}
    \item[(2.1)] For any sentence $\alpha$, the
        interpretation $^\mathcal M{\tt Hol}_\mathfrak T$
        associates to each level of the syntactical tree
        of $\alpha$ a {\em meaning\/}, represented by a
        qumix living in $\mathcal H^\alpha$ (the semantic
        space of $\alpha$);
 \item[(2.2)] Let  $ (^\alpha\mathbf O^{(k-1)}_\mathfrak
     T,\ldots, \, ^\alpha\mathbf O^{(1)}_\mathfrak T)$ be
     the $(\mathfrak T,\mathcal M^q$)-pseudo-gate tree of
     $\alpha$ and let $1 \le i < Height(\alpha)$. Then,
     $$ ^\mathcal M{\tt Hol}_\mathfrak T(Level_{i}(\alpha))= \,\,
     ^\alpha\mathbf O^{(i)}_\mathfrak
     T(^\mathcal M{\tt
     Hol}_\mathfrak T(Level_{i+1}(\alpha))).$$ In other
 words the global meaning of each level (different from
 the top level) is obtained by applying the corresponding
 pseudo-gate to the meaning of the level that occurs
 immediately above.
 \item [(2.3)] Let $Level_i(\alpha) = (\beta_1,  \ldots,
     \beta_r).$ Then: \nl  $\beta_j = \mathbf f
     \Rightarrow Red^j(^\mathcal M{\tt Hol}_\mathfrak
     T(Level_{i}(\alpha)))= \,\,^\mathfrak TP^{(1)}_0;$
     \nl $\beta_j = \mathbf t \Rightarrow Red^j({^\mathcal
     M\tt Hol}_\mathfrak T(Level_{i}(\alpha)))=
     \,\,^\mathfrak TP^{(1)}_1$, for any $j$ $(1 \le j \le
     r)$.\nl In other words, the contextual meanings of
     $\mathbf f$ and of $\mathbf t$ are always the
     $\mathfrak T$-{\em Falsity\/} and the $\mathfrak
     T$-{\em Truth\/}, respectively.
\end{enumerate}
\end{enumerate}
\end{definition}
%----------------------

On this basis, we put:
$$^\mathcal M{\tt Hol}_\mathfrak T(\alpha):=
 \,\,^\mathcal M{\tt Hol}_\mathfrak T(Level_1(\alpha)), $$
 for any sentence $\alpha$.

As an example, consider again the sentence
$$\alpha= \mathcal
K\bf{a_t}\lnot(\mathbf q \land \lnot \mathbf q) = \mathcal K
\bf{a_t}\lnot(\intercal(\mathbf q,\lnot \mathbf q,\mathbf f)).$$
As we have seen, any $(\mathfrak T, \mathcal
M^q)$-pseudo-gate-tree of $\alpha$ will have the following form:
$$\left(^\mathfrak D{\tt I}^{(1)} \otimes\,
^\mathfrak D{\tt NOT}^{(1)}_\mathfrak T \otimes \,^\mathfrak D{\tt
I}^{(1)},\,\,\, ^\mathfrak D{\tt T}_\mathfrak T^{(1,1,1)}, \,\,\,
^\mathfrak D{\tt NOT}^{(3)}_\mathfrak T,\,\,\, \mathbf
K^{(3)}_{{\mathfrak a_t}}\right).$$ Take a model
$$\mathcal M  = (T,\,Ag,\, \mathbf {EpSit}, \,den,
\, ^\mathcal M{\tt Hol})$$ such that:
\begin{itemize}
\item $^\mathcal M{\tt Hol}_{\tt
    I}(Level_{Height(\alpha)}(\alpha)= \,\, ^\mathcal M{\tt
    Hol}_{\tt I}((\mathbf q,\mathbf q,\mathbf f))\,=
    P_{\ket{\psi}}\breve{} $,\nl where
    $\ket{\psi}=\frac{1}{\sqrt{2}}(\ket{0}+ \ket{1}) \otimes
    \frac{1}{\sqrt{2}}(\ket{0}+ \ket{1})\otimes \ket{0}. $
\item $\mathbf{EpSit}$ assigns to agent $\mathfrak {a_t}$ the
    epistemic situation $$(\mathfrak
    {T_{a_t}},\,EpD_{\mathfrak {a_t}},\, \mathbf U_{{\mathfrak
 a_t}},\, \mathbf K_{{\mathfrak a_t}}),$$ where
 $EpD_{\mathfrak {a_t}}= \mathfrak D$ and $\mathbf
 K_{{\mathfrak a_t}}(\rho)= \rho$, for any $\rho \in \mathfrak
 D$. In other words, $\mathfrak {a_t}$ has a {\em maximal
 epistemic capacity\/}\footnote{See Section 3 of the first
 Part of this article.}.

\end{itemize}

We obtain:\nl $^\mathcal M{\tt Hol}_{\tt I}(\mathcal K\mathbf
a_t\lnot(\mathbf q \land \lnot \mathbf q))= P_{\ket{\varphi}}$,
 where:\nl $\ket{\varphi}=\, \mathbf
K_{\mathfrak {a_t}}^{(3)}{\tt NOT}_{\tt I}^{(3)} {\tt T}_{\tt
I}^{(1,1,1)} ({\tt I}^{(1)}\otimes {\tt NOT}_{\tt I}^{(1)} \otimes
{\tt I}^{(1)})  (\frac{1}{\sqrt{2}}(\ket{0} +\ket{1}) \otimes
\frac{1}{\sqrt{2}}(\ket{0} +\ket{1}) \otimes \ket{0})=
\frac{1}{2}(\ket{0,1,1} + \ket{0,0,1} + \ket{1,1,0} +
\ket{1,0,1}).$\nl Hence, ${\tt p}_{\tt I}(^\mathcal M {\tt
Hol}_{\tt I}(\mathcal K\mathbf a_t\lnot(\mathbf q \land \lnot
\mathbf q))) = \frac {3}{4} \neq 1$.\nl This example clearly shows
how even an agent with a maximal epistemic capacity does not
necessarily know a very simple instance of the non-contradiction
principle!

 Unlike standard {\em compositional\/}
semantics, any $^\mathcal M{\tt Hol}_\mathfrak T(\alpha)$
represents a kind of autonomous semantic context that is not
necessarily correlated with the meanings of other sentences. At
the same time, given a sentence $\gamma$, $^\mathcal M{\tt
Hol}_\mathfrak T$ determines the {\em contextual
 meaning},
 with respect to the context $^\mathcal M{\tt Hol}_\mathfrak T(\gamma)$,
 of any {\em occurrence of a subformula} $\beta$
 in the syntactical tree of $\gamma$.

%----------------------------------------------------
 \begin{definition}\label{de:newcont}
  {\em (Contextual meaning)}\nl
  Consider a sentence $\gamma$ such that
  $$Level_i(\gamma) = (\beta_{i_1},\ldots,\beta_{i_r}).  $$
  The {\em contextual meaning\/} of the occurrence $\beta_{i_j}$
  with respect to the context $^\mathcal
M{\tt Hol}_\mathfrak T(\gamma)$ is defined as follows:
$$^\mathcal
M{\tt Hol}^\gamma_\mathfrak T(\beta_{i_j}):= Red^j(^\mathcal
M{\tt Hol}_\mathfrak T(Level_i(\gamma))).   $$

   \end{definition}
%-----------------------------------

   Hence, in particular, we have for any sentence $\gamma$
$$^\mathcal M{\tt Hol}_\mathfrak T^{\gamma}(\gamma)=
 \,\,^\mathcal M{\tt Hol}_\mathfrak T(Level_1(\gamma))=
 \,\,^\mathcal M{\tt Hol}_\mathfrak T(\gamma).$$

Generally, different occurrences $\beta_{i_j}$ and $\beta_{h_k}$
of one and the same subformula $\beta$ in the syntactical tree of
$\gamma$ may receive different contextual meanings. In other
words, we may have:
$$^\mathcal M{\tt Hol}_\mathfrak T^{\gamma}(\beta_{i_j}) \,\neq\,
^\mathcal M{\tt Hol}_\mathfrak T^{\gamma}(\beta_{h_k}).
$$

When this is not the case, we will say that one is dealing with a
{\em normal model}.

  \begin{definition}{\em (Normal holistic model)} \label{de:whol}\nl
 A {\em normal holistic model} of the language $\mathcal L^{EpQC}$ is a
 holistic model  $\mathcal M$ such that for any
 truth-perspective
  $\mathfrak T$ and for any sentence $\gamma$,
  the interpretation $^\mathcal M{\tt Hol}_\mathfrak T$
determines for any occurrence $\beta_{i_j}$ of a subformula
$\beta$ of $\gamma$ in the syntactical tree of $\gamma$ the same
contextual meaning, which will be uniformly indicated by
$^\mathcal M{\tt Hol}_\mathfrak T^\gamma(\beta)$.

\end{definition}

In the following we will always refer to normal holistic models.

 Suppose that $\beta$
 is a subformula of two different formulas
 $\gamma$ and $\delta$. Generally, we have:
 $$^\mathcal M{\tt Hol}_\mathfrak T^{\gamma}(\beta)\neq\,\,
 ^\mathcal M{\tt Hol}_\mathfrak T^{\delta}(\beta).$$

 In other words, sentences may receive different contextual
 meanings in different contexts also in the case of the normal holistic semantics.

To what extent do contextual meanings and gates (associated to the
logical connectives) commute? An answer to this question is given
by the following theorem.

\begin{theorem} \label{th:commut} \nl
Consider a holistic model $\mathcal M  = (T,\,Ag,\, \mathbf
{EpSit}, \,den, \, ^\mathcal M{\tt Hol})$ and a truth-perspective
$\mathfrak T$.
\begin{enumerate}
\item[1.] Let $\lnot \alpha$ be a subformula of $\gamma$.
    Suppose that $\lnot \alpha= \beta_{i_j}$ (the formula
    occurring at the $j$-th position of the $i$-th level in
    the syntactical tree of $\gamma$), while $\alpha=
    \beta_{(i+1)_k}$. We have:\nl $^\mathcal M{\tt
    Hol}_\mathfrak T^\gamma(\lnot \alpha) = \,Red^j(^\mathcal
    M{\tt Hol}_\mathfrak T(Level_i(\gamma))) =$\nl $ \,\,
    ^\mathfrak D{\tt NOT}_\mathfrak T^{(At(\alpha))}
    (Red^k(^\mathcal M{\tt Hol}_\mathfrak
    T(Level_{i+1}(\gamma))))=$\nl $^\mathfrak D{\tt
NOT}_\mathfrak T^{(At(\alpha))}
    (^\mathcal M{\tt Hol}_\mathfrak T^\gamma(\alpha))$.
\item[2.] Let $\sqrt{id} \alpha$ be a subformula of $\gamma$.
    Suppose that $\sqrt{id} \alpha= \beta_{i_j}$, while
    $\alpha= \beta_{(i+1)_k}$. We have:\nl $^\mathcal M{\tt
    Hol}_\mathfrak T^\gamma(\sqrt{id} \alpha) = \,
    Red^j(^\mathcal M{\tt Hol}_\mathfrak
    T(Level_i(\gamma)))=$\nl $\,\,^\mathfrak D \sqrt {\tt
     I}_\mathfrak T^{(At(\alpha))}
    (Red^k(^\mathcal M{\tt Hol}_\mathfrak
    T(Level_{i+1}(\gamma))))=$\nl $\,\,^\mathfrak D\sqrt{\tt
    I}_\mathfrak
T^{(At(\alpha))}
    (^\mathcal M{\tt Hol}_\mathfrak T^\gamma(\alpha))$.
\item[3.] Let $\intercal (\alpha_1,\alpha_2,\alpha_3) $ be a
    subformula of $\gamma$. Suppose that in the syntactical
    tree of $\gamma$: $\intercal (\alpha_1,\alpha_2,\alpha_3)=
    \beta_{i_j}$, while $\alpha_1= \beta_{(i+1)_{k_1}}$,
    $\alpha_2= \beta_{(i+1)_{k_2}}$, $\alpha_3=
    \beta_{(i+1)_{k_3}}$ We have:\nl $^\mathcal M{\tt
    Hol}_\mathfrak T^\gamma(\intercal
    (\alpha_1,\alpha_2,\alpha_3)) = \,Red^j(^\mathcal M{\tt
    Hol}_\mathfrak T(Level_i(\gamma))) =$ \nl $^\mathfrak D
    {\tt T}_\mathfrak T^{(At(\alpha_1),At(\alpha_2),
    At(\alpha_3))} (Red^{k_1, k_2,k_3}(^\mathcal M{\tt
    Hol}_\mathfrak T(Level_{i+1}(\gamma))))$.
\item[4.] Let $\alpha_1 \uplus \alpha_2 $ be a subformula of
    $\gamma$. Suppose that in the syntactical tree of
    $\gamma$: $\alpha_1 \uplus \alpha_2)= \beta_{i_j}$, while
    $\alpha_1= \beta_{(i+1)_{k_1}}$, $\alpha_2=
    \beta_{(i+1)_{k_2}}$. We have:\nl $^\mathcal M{\tt
    Hol}_\mathfrak T^\gamma(\alpha_1 \uplus \alpha_2) =
    \,Red^j(^\mathcal M{\tt Hol}_\mathfrak T(Level_i(\gamma)))
    =$ \nl $^\mathfrak D {\tt XOR}_\mathfrak
    T^{(At(\alpha_1),At(\alpha_2))} (Red^{k_1, k_2}(^\mathcal
    M{\tt Hol}_\mathfrak T(Level_{i+1}(\gamma))))$.
\item[5.] Let $\mathcal U_{\mathbf a_t} \alpha$ be a
    subformula of $\gamma$. Suppose that $ \mathcal U_{\mathbf
    a_t} \alpha= \beta_{i_j}$, while $\alpha=
    \beta_{(i+1)_k}$. We have:\nl $^\mathcal M{\tt
    Hol}_\mathfrak T^\gamma(\mathcal U_{\mathbf a_t} \alpha) =
    \,Red^j(^\mathcal M{\tt Hol}_\mathfrak T(Level_i(\gamma)))
    =$\nl $\mathbf U_{\mathfrak {a_t}}^{(At(\alpha))}
    (Red^k(^\mathcal M{\tt Hol}_\mathfrak
    T(Level_{i+1}(\gamma))))= \mathbf U_{\mathfrak
    {a_t}}^{(At(\alpha))}(^\mathcal M{\tt Hol}_\mathfrak
    T^\gamma(\alpha))$.
\item[6.] Let $\mathcal K_{\mathbf a_t} \alpha$ be a
    subformula of $\gamma$. Suppose that $\mathcal K_{\mathbf
    a_t} \alpha= \beta_{i_j}$, while $\alpha=
    \beta_{(i+1)_k}$. We have:\nl $^\mathcal M{\tt
    Hol}_\mathfrak T^\gamma(\mathcal K_{\mathbf a_t} \alpha) =
    \,Red^j(^\mathcal M{\tt Hol}_\mathfrak T(Level_i(\gamma)))
    =$\nl $\mathbf K_{\mathfrak {a_t}}^{(At(\alpha))}
    (Red^k(^\mathcal M{\tt Hol}_\mathfrak
    T(Level_{i+1}(\gamma))))= \mathbf K_{\mathfrak
    {a_t}}^{(At(\alpha))}(^\mathcal M{\tt Hol}_\mathfrak
    T^\gamma(\alpha))$.

\end{enumerate}

\end{theorem}
\begin{proof} By definition of syntactical tree, of pseudo-gate
tree, of normal holistic model and of contextual meaning.

\end{proof}

Notice that, generally, the contextual meaning of a conjunction is
not the conjunction of the contextual meanings of the two members.
As a counterexample, consider the following contradictory
sentence:
$$\gamma= \mathbf q \land \lnot \mathbf q \,\,=
\intercal (\mathbf q,\lnot \mathbf q, \mathbf f ), $$ whose syntactical tree is:
$$Level_3(\gamma) = (\mathbf q,\mathbf q,\mathbf f)    $$
$$Level_2(\gamma) = (\mathbf q,\lnot \mathbf q,\mathbf f)    $$
$$Level_1(\gamma) = \intercal (\mathbf q,\lnot \mathbf q,\mathbf f)    $$

Consider a model $\mathcal M$ such that:
$$^\mathcal M{\tt Hol}_{\tt I}(Level_3(\gamma))=
P_{\frac {1}{\sqrt{2}}(\ket{0,1,0}+\ket{1,0,0})}  $$ (which is a
maximally entangled quregister with respect to the first and to
the second part\footnote{See Section 2 of the first Part of this
article.}). Hence:
$$^\mathcal M{\tt Hol}_{\tt I}(Level_1(\gamma))=\,\,
^\mathcal M{\tt Hol}_{\tt I}(\gamma)=\,\,
 ^\mathcal M{\tt Hol}_{\tt I}^\gamma(\gamma)=\,\,
P_{\frac {1}{\sqrt{2}}(\ket{0,0,0}+\ket{1,1,1})}  $$ (which is a
maximally entangled quregister). At the same time, we have:
$$^\mathcal M{\tt Hol}_{\tt I}^\gamma(\mathbf q) =
\,\,^\mathcal M{\tt Hol}_{\tt I}^\gamma(\lnot \mathbf q)\,\,
= \,\, \frac {1}{2}P^{(1)}_0 + \frac {1}{2}P^{(1)}_1,  $$
 which is a proper mixture. Consequently:
 $$^\mathfrak D{\tt T}^{(1,1,1)}
 (^\mathcal M{\tt Hol}_{\tt I}^\gamma(\mathbf q)\otimes\,\,
 ^\mathcal M{\tt Hol}_{\tt I}^\gamma(\lnot \mathbf q)\otimes\,\,
 ^\mathcal M{\tt Hol}_{\tt I}^\gamma(\mathbf f)) \,\,
 \neq \,\, ^\mathcal M{\tt Hol}_{\tt I}^\gamma(\intercal(\mathbf q,
 \lnot \mathbf q, \mathbf f)).$$
Notice that: ${\tt p}_{\tt I}(^\mathcal M {\tt Hol}^\gamma_{\tt I}
(\mathbf q \land \lnot \mathbf q)) = \frac {1}{2}$; while: \nl
${\tt p}_{\tt I}(^\mathfrak D {\tt T}^{(1,1,1,)}(^\mathcal M {\tt
Hol}^\gamma_{\tt I} (\mathbf q)\otimes \, ^\mathcal M {\tt
Hol}^\gamma_{\tt I} (\lnot \mathbf q)\otimes \, ^\mathcal M {\tt
Hol}^\gamma_{\tt I} (\mathbf f))) = \frac {1}{4}$.

\begin{definition} {(\em Harmonic and sound models)}\nl
 Let
  $\mathcal M  = (T,\,Ag,\, \mathbf {EpSit}, \,den,
                           \, ^\mathcal M{\tt Hol})$ be a model.
  \begin{itemize}
 \item  $\mathcal M$ is called {\em harmonic\/} iff the
     epistemic structure of $\mathcal M$  is harmonic, i.e.
     all agents of the structure share the same
     truth-perspective $\mathfrak T$. Furthermore, the
     interpretation-function $^\mathcal M{\tt Hol}$  is only
     defined for  the truth-perspective $\mathfrak T$.
 \item $\mathcal M$ is called {\em sound\/} iff all agents
     $\mathfrak {a_t}$ of $\mathcal M$ have  a sound epistemic
     capacity (i.e. assign the ``right'' probability-values to
     the truth-values of their
     truth-perspectives).\footnote{See Section 3 of the first
     Part of this article.}
     \end{itemize}
\end{definition}

  By {\em harmonic epistemic quantum computational semantics\/}
  ({\em sound epistemic
 quantum computational semantics\/}) we will mean the semantics
 that only refers to harmonic models (sound models).

  We can now define the notions of {\em truth\/}, {\em validity\/}
 and {\em logical consequence}.

 \begin{definition}
 {\em (Contextual truth)}\label{de:conttruth} \nl
 Let $\alpha$ be a subformula of $\gamma$.\nl
 $\vDash_{(\gamma,\mathcal M, \mathfrak T)} \alpha$  (the sentence $\alpha$
 is {\em true\/} with respect to the context $\gamma$,
 the model $\mathcal M$ and  the truth-perspective
 $\mathfrak T$) iff
 ${\tt p}_\mathfrak T(^\mathcal M{\tt Hol}^\gamma_\mathfrak T(\alpha))= 1$.

 \end{definition}

 \begin{definition}
 {\em (Truth)}\label{de:truth} \nl
 $\vDash_{(\mathcal M,\mathfrak T)} \alpha$  (the sentence $\alpha$
 is {\em true\/} with respect to the model $\mathcal M$ and
the truth-perspective
 $\mathfrak T$) iff
$\vDash_{(\alpha,\mathcal M, \mathfrak T)} \alpha$.

  \end{definition}

  Hence,
  the concept of truth turns out to be a special case of the
  concept
  of contextual truth.

  \begin{definition}{(Contextual validity)}\label{de:contvalid} \nl
  Let $\alpha$ be a subformula of $\gamma$.
  \begin{itemize}
  \item $\vDash_{(\gamma,\mathfrak T)} \alpha$  (the sentence
      $\alpha$ is {\em valid\/}  with respect to the context
      $\gamma$ and the truth-perspective $\mathfrak T$) iff
   for any model $\mathcal M$, $\vDash_{(\gamma, \mathcal M,
  \mathfrak T)} \alpha$.
  \item $\vDash_\gamma \alpha$  (the sentence $\alpha$
  is {\em valid\/} with respect to the context $\gamma$) iff
    for any truth- perspective $\mathfrak T$,
  $\vDash_{(\gamma,\mathfrak T)} \alpha$.

  \end{itemize}

  \end{definition}

 \begin{definition} {\em (Validity)} \label{de:valid}\nl

   \begin{itemize}

   \item $\vDash_{\mathfrak T} \alpha$  (the sentence $\alpha$
       is {\em valid\/} with respect to the truth-perspective
    $\mathfrak T$) iff $\vDash_{(\alpha, \mathfrak T)}
    \alpha$.
   \item $\vDash \alpha$  (the sentence $\alpha$
   is {\em valid\/})  iff $\vDash_{\alpha} \alpha$.

   \end{itemize} \end{definition}

    \begin{definition}{\em (Consequence with respect
   to a quasi-model $\mathcal M^q$)}
    \label{de:quasicons} \nl
    Let $\gamma$ be a sentence such that $\alpha$ and $\beta$ are
     subformulas of $\gamma$ and let $\mathfrak T$ be a
     truth-perspective.\nl
   \begin{itemize}
  \item $\alpha \vDash_{(\gamma,\,\mathcal M^q, \mathfrak T)}
      \beta$ (the sentence $\beta$ is a {\em consequence\/} of
      the sentence $\alpha$  with respect to the context
      $\gamma$, the quasi-model $\mathcal M^q$ and the
    truth-perspective $\mathfrak T$)    iff for any model
    $\mathcal M$ based on
     $\mathcal M^q$:
     $$\vDash_{(\gamma,\mathcal M, \mathfrak T)} \alpha\,\,\,\Rightarrow
     \,\,\,\vDash_{(\gamma,\mathcal M, \mathfrak T)} \beta.   $$

    \item
    $\alpha \vDash_{(\gamma,\mathcal M^q)} \beta$
     (the sentence $\beta$
     is a {\em  consequence\/}  of the sentence $\alpha$  with respect
     to the context $\gamma$ and the quasi-model
     $\mathcal M^q$)   iff for any  truth-perspective
      $\mathfrak T$, $\alpha \vDash_{(\gamma,\mathcal M^q,
     \mathfrak T)} \beta$.
     \end{itemize}

     \end{definition}

     \begin{definition} {\em Logical
     consequence}\label{de:logcons}\nl
     Let $\gamma$ be a sentence such that $\alpha$ and $\beta$ are
      subformulas of $\gamma$ and let $\mathfrak T$ be a
      truth-perspective.\nl
  \begin{itemize}
  \item $\alpha \vDash_{(\gamma, \mathfrak T)} \beta$ ($\beta$
      is a {\em logical consequence\/} of $\alpha$ with
      respect to the context $\gamma$ and the
      truth-perspective  $\mathfrak T$) iff for any
      quasi-model $\mathcal M^q$, $\alpha
      \vDash_{(\gamma,\mathcal M^q, \mathfrak T)} \beta$.
  \item $\alpha \vDash_{\gamma} \beta$ ($\beta$ is a {\em logical consequence\/}
  of $\alpha$ with respect to the context $\gamma$) iff for
   any truth-perspective $\mathfrak T$, $\alpha
   \vDash_{(\gamma, \mathfrak T)} \beta$.
   \item $\alpha \vDash \beta$ ($\beta$ is a {\em logical
       consequence\/}  of $\alpha$ iff for any context
       $\gamma$ such that $\alpha$ and $\beta$ are subformulas
       of $\gamma$ and $\alpha \vDash_{\gamma} \beta$.

  \end{itemize}
       \end{definition}

 The concepts of {\em consequence\/} and of {\em logical
 consequence\/}, defined above, correspond to {\em weak\/} concepts,
 being
  defined in terms of $\mathfrak T$-{\em Truth\/}, and not  in terms of the
 preorder relation $\preceq_\mathfrak T$ (as one could expect).
 From an intuitive point of view, however,  such weak notions turn
 out to be more interesting in the case of epistemic situations
 described in the framework of a holistic semantics.

Notice that  only the contextual notion of logical consequence
turns out to satisfy transitivity ($\alpha
\vDash_{\gamma}\beta\,\, \text{and}\,\,\beta
\vDash_{\gamma}\delta\, \Rightarrow \,\alpha \vDash_{\gamma}\delta
$). Full transitivity ($\alpha \vDash\beta\,\, \text{and}\,\,\beta
\vDash \delta\, \Rightarrow \,\alpha \vDash\delta $) is naturally
violated in the holistic semantics.

 As expected, in the particular case of the harmonic epistemic semantics
 (where all agents share the same truth-perspective)  the
 definitions of truth, validity and  logical consequence can be
 simplified, since the reference to $\mathfrak T$ is no longer
 necessary. Accordingly, in such a case we will write: \nl
 $\vDash_{(\gamma,\mathcal M)}^{Harm}\alpha;\,\,
 \vDash_{\mathcal M}^{Harm}\alpha$ (harmonic truth); \nl
 $\vDash_{\gamma}^{Harm}\alpha;\,\,
 \vDash^{Harm}\alpha$ (harmonic validity); \nl
  $\alpha \vDash_{(\gamma, \,\mathcal M^q)}^{Harm}\beta; \,\,
  \alpha \vDash_\gamma^{Harm}\beta;\,\,
\alpha \vDash^{Harm}\beta$   (harmonic logical consequence).

 \section{Some epistemic situations} We will now illustrate some
  significant examples of epistemic
  situations that arise in this  semantics.
 We will always refer to models   $$\mathcal M  =
 (T,\,Ag,\, \mathbf{EpSit}, \, den,
             \, ^\mathcal M{\tt Hol})$$
 such that $den(\mathbf a)= \mathfrak a; \,\, den(t)= \mathfrak t $.

 \begin{enumerate}
 \item [1)]  $\mathcal K\mathbf a_t\alpha \vDash^{Harm}
     \alpha$. \nl In the harmonic semantics, sentences that
     are known by a given agent at a given time are true.\nl
     $1)$ is an immediate consequence of the definition of
     logical consequence and of Theorem \ref{th:commut}.

    \item[2)] As a particular case of 1)  we obtain: \nl
        $\mathcal K\mathbf a_t\alpha \mathcal K\mathbf
        a_t\alpha \vDash^{Harm} \mathcal K\mathbf
        a_t\alpha$.\nl Knowing of knowing implies knowing. But
        not the other way around!

\item[3)] In the non-harmonic semantics only the two following
conditions (which are weaker than 1) and 2)) hold for any
 quasi-model $\mathcal M^q$ and any agent $\mathfrak {a_t}$ of
 $\mathcal M^q$:
\begin{enumerate}
\item[3.1)]  $\mathcal K\mathbf a_t \alpha
    \vDash_{(\mathcal M^q, \mathfrak T_{\mathfrak {a_t}})}
\alpha;  $
\item[3.2)]  $\mathcal K\mathbf a_t \alpha \mathcal
    K\mathbf a_t \alpha \vDash_{(\mathcal M^q, \mathfrak
    T_{\mathfrak {a_t}})} \mathcal K\mathbf a_t\alpha.  $
    \end{enumerate}

\item[4)] $\mathcal K\mathbf a_t\mathcal K\mathbf b_t \alpha
    \vDash^{Harm} \alpha.  $ \nl In the harmonic semantics,
    knowing that another agent knows a given sentence implies
 that the sentence in question holds. At the same time, we
 will have: \nl $\mathcal K\mathbf a_t \mathcal K\mathbf b_t
 \alpha \nvDash^{Harm} \mathcal K\mathbf a_t \alpha$.\nl Alice
  might know that Bob knows a given sentence, without knowing
 herself the sentence in question!

 \item[5)] In the harmonic sound semantics (where for any
     agent $\mathfrak {a_t}$, $\mathbf K_{\mathfrak
     a_{\mathfrak t}}\, ^{\mathfrak {T_{a_t}}}P^{(1)}_1 = \,
     ^{\mathfrak {T_{a_t}}}P^{(1)}_1$ and $\mathbf
     K_{\mathfrak a_{\mathfrak t}} \, ^{\mathfrak
     {T_{a_t}}}P^{(1)}_0 = \, ^{\mathfrak
     {T_{a_t}}}P^{(1)}_0$) we have:
 $$\vDash^{Harm} \mathcal K\mathbf a_t\mathbf t; \,\,\, \vDash^{Harm}
  \mathcal K\mathbf a_t\lnot \mathbf f.  $$ Hence, there are
 sentences that  every agent knows.

 \item[6)] $\mathcal K\mathbf a_t(\alpha \land \beta)  \nvDash
     \mathcal K\mathbf a_t\alpha;\,\, \mathcal K\mathbf
     a_t(\alpha \land \beta) \nvDash \mathcal K\mathbf
     a_t\beta. $ \nl Knowing a conjunction does not generally
     imply knowing its members.\nl

\item[7)] $\vDash_{(\gamma, \mathcal M)}\mathcal K\mathbf
    a_t\alpha\,\, \text{and} \,\, \vDash_{(\gamma,
  \mathcal M)}\mathcal K\mathbf a_t\beta\,\,\nRightarrow \,\,
   \vDash_{(\gamma, \mathcal M)}\mathcal K\mathbf a_t(\alpha
  \land \beta).$\nl Knowledge is not generally closed under
  conjunction.

\item[8)] Let $\mathcal M$ be a model  and let $\mathfrak
    {a_t}$  be an agent of $\mathcal M$.\nl
    $\nvDash_{(\mathcal M, \mathfrak T_{\mathfrak {a_t}})}
    \mathcal K\mathbf a_t(\alpha \land \lnot\alpha)$.\nl
    Contradictions are never known.\nl In order to prove 8),
    suppose, by contradiction, that there exists a model
    $\mathcal M$ and an agent  $\mathfrak {a_t} $ such
 that: $\vDash_{(\mathcal M, \mathfrak T_{\mathfrak {a_t}})}
 \mathcal K\mathbf a_t(\alpha \land \lnot\alpha)$.\nl Then, $
  {\tt p}_{\mathfrak T_{\mathfrak {a_t}}}(^\mathcal M{\tt
  Hol}_ {\mathfrak T_{\mathfrak {a_t}}}(\mathcal K\mathbf
a_t(\alpha \land \lnot \alpha)))=1 $. \nl
%Consider the
%syntactical tree of $\mathcal K\mathbf a_t(\alpha \land \lnot
%\alpha)$:
%$$\ldots \ldots  $$
%$$Level_2( \mathcal K\mathbf a_t(\alpha \land \lnot \alpha))=
%\alpha \land \lnot \alpha$$
%$$Level_1(\mathcal K\mathbf a_t(\alpha \land \lnot \alpha)) =
%K\mathbf a_t(\alpha \land \lnot \alpha)   $$
 By definition of model and by Theorem \ref{th:commut} we
have:\nl $^\mathcal M{\tt Hol}_{\mathfrak T_{\mathfrak {a_t}}}
(\mathcal K\mathbf a_t(\alpha \land \lnot \alpha)) =
%\,\,^\mathcal M{\tt Hol}_{\mathfrak T_{\mathfrak {a_t}}}
%(Level_1(\mathcal K\mathbf a_t(\alpha \land \lnot \alpha))=$
 \mathbf K_{\mathfrak a_\mathfrak t} (^\mathcal M{\tt
Hol}_{\mathfrak T_{\mathfrak {a_t}}}(\alpha \land \lnot
\alpha)) $. \nl Consequently, by hypothesis, ${\tt
p}_{\mathfrak T_{\mathfrak {a_t}}}( \mathbf K_{\mathfrak
a_{\mathfrak t}} (^\mathcal M{\tt Hol}_{\mathfrak T_{\mathfrak
{a_t}}}(\alpha \land \lnot \alpha))) =1 $. \nl Thus, by the
properties  of knowledge operations: \nl ${\tt p}_{\mathfrak
T_{\mathfrak {a_t}}} (^\mathcal M{\tt Hol}_{\mathfrak
T_{\mathfrak {a_t}}}(\alpha \land \lnot \alpha)) =1 $, which
is impossible, owing to the following Lemma (of the holistic
semantics).

\begin{lemma}\label{le:nocontrad}
 For any sentence $\alpha$, for any truth-perspective
  $\mathfrak T$ and for any holistic model $\mathcal M$,
  $${\tt p}_\mathfrak T(^\mathcal M{\tt Hol}_
  {\mathfrak T}(\alpha \land \lnot \alpha)) \neq 1.   $$

\end{lemma}

 \item[9)] In the non-harmonic semantics the following
     situation is possible:\nl $ \vDash_{(\mathcal M,\mathfrak
     T_{\mathfrak {a_t}})} \mathcal K\mathbf a_t \mathcal
     K\mathbf b_t \mathbf f.$ \nl In other words, according to
     the truth-perspective of Alice it is true
  that Alice (at time $t$) knows that Bob (at time $t$) knows
   the {\em Falsity\/} of Alice's truth-perspective. \nl As an
   example, consider a (non-harmonic)  model $\mathcal M$ with
   two agents $\mathfrak {a_t}$ and $\mathfrak {b_t}$
   satisfying the following conditions:
    \begin{enumerate}
    \item[a)] the epistemic distance between the
        truth-perspectives of $\mathfrak {a_t}$ and of
        $\mathfrak {b_t}$ is greater than or equal to
        $\frac {1}{2}$.\footnote{The concept of {\em
        epistemic distance\/} has been defined in Section
        2 of the first Part of this article.} In such a
        case we have:  $^{\mathfrak T_{\mathfrak
  {a_t}}}P_1^{(1)} \,\,
      \preceq_{\mathfrak T_{\mathfrak {b_t}}}
      \,\,^{\mathfrak T_\mathfrak {a_t}}P_0^{(1)}$ \nl
      (according to Bob's truth-perspective, Alice' s
      {\em Truth\/} precedes Alice's {\em Falsity\/});
   \item[b)] $\mathbf K_{\mathfrak {b_t}} \,^{\mathfrak
       T_{\mathfrak {a_t}}} P_0^{(1)} = \,\, ^{\mathfrak
       T_{\mathfrak {a_t}}}P_1^{(1)}   $\nl (the
       information according to which Bob knows Alice's
       {\em Falsity\/} is true with respect to Alice's
       truth-perspective);
   \item [c)] $\mathbf K_{\mathfrak a_{\mathfrak t}}
       \,^{\mathfrak T_{\mathfrak {a_t}}}P_1^{(1)} = \,\,
       ^{\mathfrak T_{\mathfrak {a_t}}}P_1^{(1)}   $\nl
       (Alice at time $\mathfrak t$ has a sound epistemic
       capacity).

   \end{enumerate}

   Consider the syntactical tree of $\mathcal K\mathbf
   a_t\mathcal K\mathbf b_t \mathbf f$:
   $$ Level_3(\mathcal K\mathbf a_t\mathcal K\mathbf b_t \mathbf f) = \mathbf f $$
   $$ Level_2(\mathcal K\mathbf a_t\mathcal K\mathbf b_t \mathbf f) = K\mathbf b_t \mathbf f $$
   $$ Level_1(\mathcal K\mathbf a_t\mathcal K\mathbf b_t \mathbf f) =
   \mathcal K\mathbf a_t \mathcal K\mathbf b_t\mathbf f $$

   The qumixes assigned by $^\mathcal M{\tt Hol}_{\mathfrak
T_{\mathfrak {a_t}}}$ to the levels of this tree are:

$^\mathcal M{\tt Hol}_{\mathfrak T_{\mathfrak {a_t}}}
(Level_3(\mathcal K\mathbf a_t\mathcal K\mathbf b_t \mathbf
f)) = \,\, ^{\mathfrak T_{\mathfrak {a_t}}}P_0^{(1)}$ \nl (by
definition of model);\nl $^\mathcal M{\tt Hol}_{\mathfrak
T_{\mathfrak {a_t}}} (Level_2(\mathcal K\mathbf a_t\mathcal
K\mathbf b_t \mathbf f)) = \mathbf K_{\mathfrak b_{\mathfrak
t}} (^\mathcal M{\tt Hol}_{\mathfrak T_{\mathfrak {a_t}}}
(Level_3(\mathcal K\mathbf a_t\mathcal K\mathbf b_t \mathbf
f))= \,\, ^{\mathfrak T_\mathfrak {a_t}}P_1^{(1)} $ \nl (by
definition of model and by b));\nl $^\mathcal M{\tt
Hol}_{\mathfrak T_{\mathfrak {a_t}}} (Level_1(\mathcal
K\mathbf a_t\mathcal K\mathbf b_t \mathbf f)) = \mathbf
K_{\mathfrak b_{\mathfrak t}} (^\mathcal M{\tt Hol}_{\mathfrak
T_{\mathfrak {a_t}}} (Level_2(\mathcal K\mathbf a_t\mathcal
K\mathbf b_t \mathbf f))= \,\, ^{\mathfrak T_\mathfrak
{a_t}}P_1^{(1)} $ \nl (by definition of model and by c)). \nl
Hence, ${\tt p}_{\mathfrak T_{\mathfrak {a_t}}} (^\mathcal
M{\tt Hol}_{\mathfrak T_{\mathfrak {a_t}}} (\mathcal K\mathbf
a_t\mathcal K\mathbf b_t \mathbf f)) = 1$ and
$\vDash_{(\mathcal M,\mathfrak
       T_{\mathfrak {a_t}})}\mathcal K\mathbf a_t\mathcal
       K\mathbf b_t \mathbf f $.
\nl Notice that
$$\vDash_{(\mathcal M,\mathfrak T_{\mathfrak {a_t}})}\mathcal K\mathbf
a_t\mathcal K\mathbf b_t \mathbf f \,\,\,\nRightarrow \,\,\,
\vDash_{(\mathcal M,\mathfrak T_{\mathfrak {b_t}})}\mathcal
K\mathbf b_t
\mathbf f.$$ In other words, the following situation is
possible:  \
\begin{itemize}
\item According to Alice's truth-perspective, it is true that Alice knows
that Bob knows the {\em Falsity}.
\item However, according to Bob's truth-perspective it is not true that Bob knows
the {\em Falsity}.

\end{itemize}
Roughly, we might  say: Alice knows that Bob is wrong.
However, Bob is not aware of being wrong!

  \end{enumerate}

  The epistemic situations illustrated above seem to reflect
  pretty well some  characteristic limitations of the
  real processes of  acquiring information and knowledge. Owing to
  the limits of epistemic domains, understanding and knowing
  are not generally closed under logical consequence.
   Hence, the unpleasant   phenomenon  of {\em logical
   omniscience\/}
   is here avoided. We have, in
  particular, that knowledge is  not generally closed under
  logical conjunction, as in fact happens in the  case of concrete  memories
  both of human and of artificial intelligence. It is also admitted
  that an agent can understand (or know) a conjunction, without
  being able to understand (to know) its members.
  Such situation, which might appear {\em prima facie\/}
 somewhat ``irrational'', seems to be instead  deeply in agreement with
 our  use of natural languages, where sometimes
 agents show to use correctly and to understand some {\em global\/}
 expressions
 without being able to understand their (meaningful) parts.

\end{document}